\documentclass[aps,pra,twocolumn,reprint]{revtex4-1} 
\usepackage{graphicx}

\usepackage{dcolumn}
\usepackage{ytableau}

\usepackage{bm}

\usepackage{epsfig} 
\usepackage{epsf} 
\usepackage{pdfpages}
\usepackage{amsmath} 
\usepackage{amssymb} 
\usepackage{amsfonts} 
\usepackage{amscd} 
\usepackage{amsthm} 
\usepackage{wasysym} 
\usepackage{listings} 
\usepackage{cases} 
\usepackage{url} 
\usepackage{times} 
\usepackage[utf8]{inputenc} 
\usepackage{color}
\usepackage{setspace}

\newcommand{\be}{
\begin{equation}
	} 
	\newcommand{\ee}{
\end{equation}
} 
\newcommand{\bfg}{
\begin{figure}
	} 
	\newcommand{\efg}{
\end{figure}
} 
\newcommand{\bra}[1]{\langle#1|} 
\newcommand{\ket}[1]{|#1\rangle}

\newcommand{\Hcal}{\mathcal{H}} 
\newcommand{\B}{\mathcal{B}} 
\newcommand{\C}{\mathbb{C}} 
\newcommand{\lpn}{\lambda \dashv n}

\newcommand{\la}{\lambda}
\newcommand{\las}{\bar{\lambda}}
 
\newcommand{\floor}[1]{\lfloor #1 \rfloor} 
\newcommand{\Tr}{\text{Tr}} 
\newcommand{\ud}{\frac{1}{2}} 
\newcommand{\mean}[1]{\big\langle #1 \big\rangle} 
 
\newcommand{\TF}[1]{\mathcal{F}\Big(#1\Big)} 
\newcommand{\TFm}[1]{\mathcal{F}^{-1}\Big(#1\Big)}

\newcommand{\hooksA}{ \ytableausetup{boxsize=1.1em} 
\begin{ytableau}
	4 & 3 & 2 & 1 \\
\end{ytableau}
} 
\newcommand{\hooksB}{ \ytableausetup{boxsize=1.1em} 
\begin{ytableau}
	4 & 2 & 1 \\
	1 \\
\end{ytableau}
} 
\newcommand{\hooksC}{ \ytableausetup{boxsize=1.1em} 
\begin{ytableau}
	3 & 2 \\
	2 & 1\\
\end{ytableau}
} 
\newcommand{\hooksD}{ \ytableausetup{boxsize=1.1em} 
\begin{ytableau}
	4 & 1 \\
	2 \\
	1 \\
\end{ytableau}
} 
\newcommand{\hooksE}{ \ytableausetup{boxsize=1.1em} 
\begin{ytableau}
	4\\
	3 \\
	2 \\
	1 \\
\end{ytableau}
}

\newcommand{\dsym}{ \ytableausetup{boxsize=0.4em} 
\begin{ytableau}
	\, & \, & \none[\cdot] & \none[\cdot] & \none[\cdot] & \, \\
\end{ytableau}
}

\newcommand{\deven}{ \ytableausetup{boxsize=0.4em} 
\begin{ytableau}
	\, & \, & \none[\cdot] & \none[\cdot] & \none[\cdot] & \, \\
	\, & \, & \none[\cdot] & \none[\cdot] & \none[\cdot] & \, \\
\end{ytableau}
}

\newcommand{\dodd}{ \ytableausetup{boxsize=0.4em} 
\begin{ytableau}
	\, & \, & \none[\cdot] & \none[\cdot] & \none[\cdot] & \, & \, \\
	\, & \, & \none[\cdot] & \none[\cdot] & \none[\cdot] & \, \\
\end{ytableau}
}
\newcommand{\dhalf}{ \ytableausetup{smalltableaux}
\begin{ytableau}
	0 & 0 & \none[\,] & \none[\cdots] & \none[\,] & 0 \\
	1 & 1 & \none[\,] & \none[\cdots] & \none[\,] & 1 \\
\end{ytableau}
}

\ytableausetup{boxsize=0.6em}

\newtheorem{myprop}{Proposition}

\newtheorem{mytheorem}{Theorem}

\begin{document} \sloppy 
\title{All possible permutational symmetries of a quantum system} 
\author{Ludovic Arnaud} \affiliation{13 lotissement le couserans, 09100 La Tour-du-Crieu, France.}
\email{ludovic.p.arnaud@gmail.com}
\begin{abstract}
We investigate the intermediate permutational symmetries of a system of qubits, that lie in between the perfect symmetric and antisymmetric cases. We prove that, on average, pure states of qubits picked at random with respect to the uniform measure on the unit sphere of the Hilbert space are almost as antisymmetric as they are allowed to be. We then observe that multipartite entanglement, measured by the generalized Meyer-Wallach measure, tends to be larger in subspaces that are more antisymmetric than the complete symmetric one. Eventually, we prove that all states contained in the most antisymmetric subspace are relevant multipartite entangled states in the sense that their 1-qubit reduced states are all maximally mixed.
\end{abstract}
\maketitle
\date{\today} 
\section{Introduction}
Understanding how information is stored in a quantum system and how it can be extracted is one of the main goals of quantum information science. Because quantum mechanics is often counter-intuitive, this goal is as challenging as it is promising. Historically, the existence of quantum superpositions and the interference they imply were the first aspects of quantum mechanics that confronted  our intuition. When we considered measurements of an individual spin $\frac{1}{2}$ in the vertical direction, the states ${\ket \uparrow}$ and ${\ket \downarrow}$ were easy to interpret classically. However, superposition of states like $({\ket \uparrow}+{\ket \downarrow})/\sqrt{2}$ were puzzling and the statistical interpretation on a lot of copies was the only resort. Nowadays, this superposition is seen just as classical as the {\em up} and {\em down} states. We just rename it ${\ket \rightarrow}$ and consider that it only makes sense to measure it in the horizontal $x$ direction. Performing the measurement in the vertical direction is possible, but it will not give any information at all. It will disturb the system so much that it will be brought in as one of the ``vertical" states with perfect probability.

Then quantum entanglement came into the play \cite{Schroedinger1935,EPR35,Bell1964,Aspect1981} and challenged our intuition even more. The essence of entanglement is well summarized by considering the so-called {\em bipartite entanglement}. Such kinds of entanglement states that the information about a quantum system is not only encoded exclusively in its parts, but it is also encoded in the correlations between the parts. Remarkably, when a bipartite quantum system is maximally entangled, the information appears to be fully encoded in these correlations and no longer in the system’s constituents. Because the different parts of a whole system are located at different spatial positions, bipartite entanglement contradicts local realism. Bipartite entanglement is well understood nowadays. Next comes the question of entanglement when the number of parties is bigger than two, the so called {\em multipartite entanglement} . Without any surprises, multipartite entanglement is much richer than bipartite entanglement, and thus more difficult to understand \cite{Horodecki2009}. It leads to stronger contradictions with local realism than bipartite entanglement \cite{Greenberger1990} and several inequivalent classes of entangled states exist as soon as three qubits are considered \cite{Dur2000}. Multipartite entanglement is also central in several applications like one-way quantum computing \cite{Raussendorf2001}. Its dynamics has revealed a surprisingly large variety of flavors when exposed to a dissipative environment \cite{Aolita2008,Barreiro2010}.

Recently, a lot of work has focused on particular kinds of qubit states, totally invariant under permutation of their qubits. This kind of state is really interesting because they are analytically tractable and easy to work with numerically. They exhibit high entanglement content, especially in terms of their geometric entanglement \cite{Martin2010,Aulbach2010,Markham2011,Baguette2014}, non-local behavior \cite{Wang2012,Wang2013}, convenient representation \cite{Giraud2015} and involvement in experimental setups \cite{Toth2007,Kiesel2007,Thiel2007,Prevedel2009,Wieczorek2009,Chiuri2010}. However, in some aspects, the power of the permutational symmetry is also a weakness. It is a strongly constraining symmetry that a lot of interesting quantum states do not satisfy for more than three qubits, particularly the states that are known for their high entanglement content relative to different kinds of measures \cite{Laflamme1996,Higuchi2000,Brierley2007,Gour2010}. In \cite{Arnaud2013}, it is also demonstrated that a symmetric state of qubits cannot have its reduced states all maximally mixed, except in the case where those reduced states are the smallest possible, i.e., with only one qubit each. For symmetric states, reduced states formed with a pair, triplet, and etc. will never be all maximally mixed. That is surprising because in a given Hilbert space it is always possible to find states with all their reductions that keep about 18\% of the total number of qubits maximally mixed \cite{Arnaud2013}. In the context of quantum error correction, such states are therefore robust to the loss of about 18\% of their qubits because they do not encode information.

For these reasons, it is then quite natural to explore beyond the perfect permutational symmetry by still capturing some of its aspects that make it so convenient. To get some intuition on how to do such a thing, let us consider two qubits seen as two spins $\frac{1}{2}$. It is well known that arbitrary states of such a system are linear combinations of the symmetric components (formed with the three triplets) and an antisymmetric component (formed with the singlet). For more than two qubits, the situation becomes richer because a given qubit can have a symmetric relationship with some qubits and an antisymmetric relationship with some others. A 3-qubit state will decompose in different parts: one part will be indeed a totally symmetric part completely invariant by any permutation of its quits. Then, another part will have symmetry between qubits 1 and 2 but antisymmetry between 1 and 3 and 2 and 3. In the same vein, an other part will have antisymmetry between 1 and 2 but symmetry between 1 and 3 and 2 and 3. Obviously, there will be four other parts corresponding to the four other ways to fix the symmetries between the qubits. Note that, however, it will not contain a completely antisymmetric part because it is impossible to antisymmetrize more than two qubits. The goal of this paper is to study quantum states that have these kinds of intermediate symmetries that lie in between the perfect symmetric and antisymmetric ones. Those symmetries will be described thanks to the formalism of the representations of the symmetric group.

The layout of this paper is the following: Sec. II introduces important notions about the symmetric group and its representations and rigorously defines the intermediate symmetries. In Sec. III, a measure of the amount of intermediate symmetries contained in a quantum state is introduced. Then, all the moments and the probability distribution of that measure are calculated analytically for random states of qubits and confronted with numerical simulation. This statistical analysis shows that the maximal antisymmetry is a generic feature of random quantum states. In Sec. IV, some relationships between quantum states with intermediate symmetries and entanglement are considered. A natural extension of the Majorana representation \cite{Majorana1932} is introduced and some numerical calculations are presented. At the end of the section, a theorem is derived showing that the most antisymmetric states are a special kind of maximally entangled states in the sense that no information is contained in their individual qubits. In other words, those states are quantum error correcting codes robust against any 1-qubit errors. Finally, some conclusions are drawn in Sec. V.

\section{Representation theory of the symmetric group} 
In this part, we will introduce the mathematical background that is going to be used throughout this article. The minimum amount of required concepts and notations will be described briefly and relevant formulas will be written without demonstration. A more rigorous approach with more details can be found in many reference books concerning the symmetric group and its representations \cite{Goodman2009,Boerner1970,Serre1977,Wilton1991}.
 
\subsection{Permutations and cycle notation}
The symmetric group of $n$ elements, noted $S_n$, is the set of permutations of those elements where the composition of permutations plays the role of the group multiplication. There are $n!$ such permutations. Some of these permutations are called {\em cycles} because they exchange the elements of a subset in a circular fashion. The number of elements that are permuted in a given cycle is called the {\em length} of the cycle. For instance, the permutation that transforms the string $\{1,2,3,4,5\}$ to the string $\{3,1,2,4,5\}$ can be seen as a cycle on the subset $\{1,2,3\}$ in the sense that 1 goes to position 2, 2 goes to position 3 and 3 goes to position 1. At the same time, 4 and 5 keep their position. Therefore, it is a cycle of length 3 or similarly a 3-cycle.

Cycles can be written in {\em cycle notation} where the elements are written in between parentheses such that the first element in the parentheses goes to the position of the second, the second goes to the position of the third one, and so on until the last element takes the position of the first one. Like this, the previous permutation is written $(123)$.

What makes cycles interesting is that any permutation can be constructed as a combination of disjoint cycles. The permutation that transforms the string $\{1,2,3,4,5\}$ to the string $\{3,1,2,5,4\}$ can be seen as two cycles, one on the subset $\{1,2,3\}$ and one on the subset $\{4,5\}$, written $(123)(45)$. Note that, in the cycle notation, elements that do not change position, i.e., cycles of size one, are omitted. That is why the permutation considered above that could be written $(123)(4)(5)$ is simply written to $(123)$. Two permutations are said to have the same cycle structure if they are constructed thanks to cycles with identical lengths. For instance, the permutation $(123)(34)$ and $(145)(23)$ are a combination of a 3-cycle and a 2-cycle. 

\subsection{Irreducible representations of $S_n$, partitions, and characters}
There is an interesting fact concerning the symmetric group. Its irreducible representations are in direct correspondence with its conjugacy classes. Moreover, it is possible to show that permutations that belong to the same conjugacy class have the same cycle structure. Therefore, each irreducible representation can be labeled by a quantity which reminds the cycle structure, that is to say a partition $\lambda$ of the number of elements $n$, noted $\lambda \dashv n$. A partition $\lambda$ can be written thanks to a {\em partition vector}, which is a vector with monotonic decreasing entries that sums to $n$,
\begin{eqnarray}
	\lambda&=&(\lambda_1,\cdots,\lambda_n) \text{ with }\lambda_1\ge \cdots \ge \lambda_n \text{ and } \sum_{i=1}^n \lambda_i=n. \quad 
\end{eqnarray}
Another way to represent a partition $\lambda$ can be done thanks to a {\em Young diagram}, a diagram that is constructed by gluing together rows of $\lambda_i$ square boxes from top to bottom, for $i$ in $[1,n]$. As an example for $n=4$, all the partition vectors and their corresponding Young diagrams are given by 
\begin{eqnarray}
	(4,0,0,0) &\equiv& \ydiagram{4}, \, (3,1,0,0) \equiv \ydiagram{3,1}, \, (2,2,0,0) \equiv \ydiagram{2,2},\nonumber\\
	(2,1,1,0) &\equiv& \ydiagram{2,1,1} \, \text{ and } (1,1,1,1) \equiv \ydiagram{1,1,1,1}.\nonumber 
\end{eqnarray}
Partitions are listed here in the {\em inverse lexicographical order}, which means that $\lambda>\lambda'$, if and only if the first non-zero difference $(\lambda_i - \lambda'_i) > 0$ for $i$ in $[1,n]$. The position of a partition with respect to this ordering will be used in some formulas by simply writing it $\lambda$, the context preventing any risk for confusion. Like this, $(4,0,0,0) \equiv 1$ and $(3,1,0,0) \equiv 2$.

From partitions and Young diagrams, several relevant quantities need to be defined. The multiplicity of $i$ of a partition $\lambda$ shorthanded with an exponent index as $\lambda^i$ counts the number of times the value $i$ appears in a partition. For instance, if $\lambda=(3,1,0,0)$, then $\lambda^1=\lambda^3=1$ and $\lambda^2=\lambda^4=0$. The {\em hooks} of a box in a given Young diagram is defined as the set of boxes that are below and to the right of the box, including the considered box itself. The {\em hook-length} of a given box is simply the total number of boxes in its hook. As an example, here are all the Young diagrams for $n=4$ with the values of the hook-length of each box: 
\begin{equation}
	\hooksA, \, \hooksB, \, \hooksC, \, \hooksD, \, \hooksE\label{hooklengths}.\ytableausetup{boxsize=0.6em} 
\end{equation}

Note that we will denote partitions relative to the irreducible representations of the symmetric group with greek indices starting at the letter $\lambda$. Greek indices starting at the letter $\rho$ will be used to represent a given conjugacy class. By extension, any permutations that belong to the same conjugacy class will also be written $\rho$. It is then useful to evaluate the number of permutation in the class $\rho$ as 
\begin{equation}
	|\rho|=\frac{n!}{\prod_{i=1}^n(\rho^i !)i^{\rho^i}}.\label{rhonumber} 
\end{equation}

The character of the permutation $\pi$ in the representation $\lambda$ noted $\chi_{\la}(\pi)$ is a number defined as the trace of the matrix of the permutation $\pi$ in the representation $\lambda$. It forms a {\em class function} in the sense that it has the same value for all the permutations that belong to the same conjugacy class. For that reason, it is common to only consider the values of the character written $\chi_{\la}(\rho)$ associated with the conjugacy class $\rho$ and arrange them together in the matrix with element $\chi_{\lambda \rho}$, called the {\em character table}. In the Appendix, the so-called Frobenius’ formula is presented, which allows us to calculate all those characters. Characters satisfy the two following orthogonality relations:
\begin{eqnarray}
	\sum_{\pi \in S_n} \chi_{\la}(\pi) \chi_{\lambda'}(\pi) &=& n!\,\delta_{\lambda \lambda'}\label{chi1},\\
	\sum_{\lambda \dashv n} \chi_{\la}(\rho) \chi_{\la}(\rho') &=& \frac{n!}{|\rho|}\,\delta_{\rho \rho'}\label{chi2}. 
\end{eqnarray}

\subsection{Schur-Weyl duality} The central concept of this article is called the Schur-Weyl duality. It states that the finite dimension Hilbert space of $n$ qudits $(\C^d)^{\otimes n}$ decomposes as a direct sum of orthogonal subspaces $\Hcal_{\la}$, 
\begin{equation}
	(\mathbb{C}^d)^{\otimes n} \simeq \bigoplus_{\lpn} \Hcal_{\la}. 
\end{equation}
Each $\Hcal_{\la}$ is constructed as the tensor product of the irreducible representation of the special unitary group SU$(d)$ and the irreducible representation of the symmetric group $S_n$, both labeled by the partition $\lambda$. The dimension $f_\lambda$ of the irreducible representations of SU$(d)$ happens to be inversely proportional to the product of all the hook-lengths of $\lambda$ noted $h_{\la}$, 
\begin{equation}	
	f_{\la}=\frac{n!}{h_{\la}}. 
\end{equation}
The dimension $d_\lambda$ of the irreducible representation of the symmetric group $S_n$ can be obtained thanks to the following formula applied to the partition vector $\lambda$: 
\begin{equation}
	\label{dimSn} d_\lambda=\prod_{i<j}^{d}\frac{\lambda_i-\lambda_j+j-i}{j-i}. 
\end{equation}
The dimension $D_{\la}$ of each $\Hcal_{\la}$ is therefore given by the product $f_{\la}d_{\la}$ and it clearly satisfies $\sum_{\la \dashv n}D_\la=d^n$.

To get some intuition about the structure of the subspaces $\Hcal_\la$, it is necessary to understand how the states they contain are constructed. Basically, the states of a possible basis for $\Hcal_\la$ are built by symmetrizing and antisymmetrizing the states of the computational basis following the pattern given by the Young diagram $\lambda$: each row will correspond to a symmetrization and each column will correspond to an antisymmetrization. To visualize the states that spanned $\Hcal_\la$, it is helpful to take each state of the computational basis state, write their arguments as binary indices in the Young diagram, and rearrange the indices according to all the possible permutations. Because of the antisymmetrization, the only states that are going to contribute are the ones leading to permutations of the filled Young diagram with each column containing distinct values. For instance, the subspace labeled by the diagram $\ytableausetup{boxsize=0.6em}\ydiagram{2,1}$ will not be spanned by the state $\ket{000}$. Clearly there is no permutation of $
\ytableausetup{smalltableaux}\begin{ytableau}
0 & 0 \\
0 \\
\end{ytableau}
$ that satisfies the column constrain. On the other hand, the states $\ket{001}$, $\ket{010}$ and $\ket{100}$ will span the subspace because they can be rearranged in the form $
\ytableausetup{smalltableaux}\begin{ytableau}
0 & 0 \\
1 \\
\end{ytableau}
$
compatible with the antisymmetrization. As a consequence of such a construction, it is impossible to antisymmetrize more than $d$ qudits. Therefore, the subspaces labeled by the Young tables that contains columns with more than $d$ boxes are zero dimensional.

To make such a construction more systematic, it is convenient to introduce the set of projectors $P_{\la}$ on each $\Hcal_{\la}$. They can be constructed as 
\begin{equation}
	\label{proj} P_{\la}=\frac{1}{h_{\la}}\sum_{\pi \in S_n} \chi_{\la}(\pi) U_\pi. 
\end{equation}
$U_{\pi}$ is the unitary operator that maps the permutation $\pi$ in the Hilbert space. Its matrix elements are given by 
\begin{equation}
	\label{P2U} (U_{\pi})_{i,j}=\delta_{i\pi(j)}, 
\end{equation}
where $i$ and $j$ are $n$ entries $d$-valued vectors that index the computational basis. By definition, $\Tr(P_{\la})=D_{\la}$ and from Eq. (\ref{chi1}), it is straightforward to verify that the projectors $P_\lambda$ form a set of orthogonal projectors such that 
\begin{equation}
	\label{orthoproj} P_{\la}P_{\lambda'}=\delta_{\lambda \lambda'}P_{\lambda'}. 
\end{equation}
From each $P_{\la}$, an orthonormal basis for each $\Hcal_{\la}$ can be built by applying a Gram–Schmidt process to the set of linear independent columns $(P_{\la})_i$. Such a basis will be noted $\B_{\la}$ with vectors $\ket{b^{\la}_k}$, for $k\in[1,D_{\la}]$. Note that the basis $\B_{(n,0,\cdots,0)}$ of the complete symmetric situation are nothing other that the Dicke states \cite{Dicke1954} and the other basis vector can be seen as their generalization.

\subsection{Example for two qubits} To see how the previous formalism works, let us consider the 2-qubit case. The parameters are thus $n=2$ and $d=2$. The integer $2$ can be decomposed as $2+0$ and $1+1$. Therefore, the Hilbert space decomposes in two subspaces as 
\begin{equation}
	\C^2 \otimes \C^2 \simeq \Hcal_{(2,0)} \oplus \Hcal_{(1,1)}. 
\end{equation}
To construct the projectors $P_{(2,0)}$ and $P_{(1,1)}$, we first need to express all the relevant quantities concerning $S_2$. This group contains only two elements: the identity element written $e$ and the transposition $(12)$. The corresponding unitary operators can be calculated from Eq. (\ref{P2U}) as $U_{e}$ simply being the identity operator and $U_{(1,2)}$ being nothing other than the SWAP operator. In matrix form 
\begin{equation}
	U_{e}=
	\begin{pmatrix}
		1 & 0 & 0 & 0 \\
		0 & 1 & 0 & 0 \\
		0 & 0 & 1 & 0 \\
		0 & 0 & 0 & 1 \\
	\end{pmatrix}
	\text{ and } \,U_{(12)}=
	\begin{pmatrix}
		1 & 0 & 0 & 0 \\
		0 & 0 & 1 & 0 \\
		0 & 1 & 0 & 0 \\
		0 & 0 & 0 & 1 \\
	\end{pmatrix}
	.\nonumber 
\end{equation}
The character table of $S_2$ can be calculated for instance thanks to the Frobenius’ formula (see the Appendix) 
\begin{equation}
	\begin{array}{|c | c c |}
		\hline
		\rho \diagdown \lambda & (2,0) & (1,1) \\
		\hline
		e & 1 & 1 \\
		(12) & 1 & -1 \\
		\hline
	\end{array}
	\nonumber 
\end{equation}
and from the definition of the hook-length, one can calculate that $h_{(2,0)}=h_{(1,1)}=2$. Equation (\ref{proj}) gives the projectors 
\begin{eqnarray}
	P_{(2,0)}&=&\frac{1}{2}(U_{e}+U_{(12)})=
	\begin{pmatrix}
		1 & 0 & 0 & 0 \\
		0 & \ud & \ud & 0 \\
		0 & \ud & \ud & 0 \\
		0 & 0 & 0 & 1 \\
	\end{pmatrix},
	\\
	P_{(1,1)}&=&\frac{1}{2}(U_{e}-U_{(12)})=
	\begin{pmatrix}
		0 & 0 & 0 & 0 \\
		0 & \ud & -\ud & 0 \\
		0 & -\ud & \ud & 0 \\
		0 & 0 & 0 & 0 \\
	\end{pmatrix}. 
\end{eqnarray}
Taking the traces of the operators gives the dimension of the subspaces 
\begin{eqnarray}
	D_{(2,0)}&=&\Tr(P_{(2,0)})=3\nonumber\\
	D_{(1,1)}&=&\Tr(P_{(1,1)})=1. 
\end{eqnarray}
The three linear independent columns of $P_{(2,0)}$ and the only distinct columns of $P_{(1,1)}$ already form an orthogonal basis. After normalization we obtain the basis 
\begin{eqnarray}
	\B_{(2,0)}&=&\{\ket{00},\frac{\ket{01}+\ket{10}}{\sqrt{2}},\ket{11}\}\label{symbase2},\\
	\B_{(1,1)}&=&\{\frac{\ket{01}-\ket{10}}{\sqrt{2}}\}.\label{antisymbase2} 
\end{eqnarray}

Such a decomposition is quite natural in the context of 1/2 spins. $\Hcal_{(2,0)}$ corresponds to the {\em symmetric subspace} and its basis states are the triplet states. $\Hcal_{(1,1)}$ corresponds to the {\em antisymmetric subspace} that contains the unique singlet states. It is interesting to notice that this subspace is maximally entangled because the singlet states is a Bell state.
\ytableausetup{boxsize=0.4em}
When the number of qudits is bigger, arbitrary partitions lead to {\em intermediate symmetry}. Those intermediate symmetries mean that a given qudit will be symmetrized with respect to a group of qubits and antisymmetrized with another group. To be more precise, a given intermediate symmetry will be called $\la$-symmetry and states that will belong to an individual subspace $\Hcal_\la$ will be called $\la$-symmetric states. Studying those states in the context of quantum information is the point of the next two parts.

\section{Qubits states and $\lambda$-symmetry} 
\subsection{Generalities} From now on, we will only focus on the qubits case ($d=2$) and write the size of the whole Hilbert space $N=2^n$. It simplifies a lot of the formulas presented in the previous part. One simplification comes from the fact that it is impossible to antisymmetrize more than two qubits. In terms of Young diagrams, it means that only diagrams with no more than two rows will lead to non-zero dimensional subspaces. Starting from the diagram $\overbrace{\dsym}^{n}$ , the most ``antisymmetric" partition will be represented by the diagram $\overbrace{\deven}^{n/2}$ if $n$ is even, and $\overbrace{\dodd}^{n/2 + 1}$ if $n$ is odd. In total there are $\floor{\frac{n}{2}}+1$ such diagrams. To calculate the dimension $D_\la$, the hook-lengths of those diagrams need to be calculated. In terms of the index of the partition with respect to the inverse lexicographical order $\la$, the number of boxes in the first row is given by $n-\la + 1$. The number of boxes in the second row is given by $\la - 1$. The hook-lengths are thus all given according to the pattern 
\begin{equation}
	\label{qubithook} \ytableausetup{boxsize=3.4em} 
	\begin{ytableau}
		\underset{+1}{\scriptstyle n - \la + 1} & \underset{+1}{\scriptstyle n - \la} & \none[\hspace{-0.5cm}\dots]\\
		\scriptstyle \la -1 & \scriptstyle \la -2 & \none[\hspace{-0.5cm}\dots]\\
	\end{ytableau}
	\hspace{-0.5cm} 
	\begin{ytableau}
		\underset{+1}{\scriptstyle n - 2\la + 3} & \scriptstyle n - 2\la + 2 & \none[\hspace{-0.5cm}\dots]\\
		\scriptstyle 1 \\
	\end{ytableau}
	\hspace{-0.5cm} 
	\begin{ytableau}
		\scriptstyle 2 & \scriptstyle 1 \\
	\end{ytableau}
	.\nonumber \ytableausetup{boxsize=0.6em} 
\end{equation}
The product of all those hook-lengths is 
\begin{eqnarray}
	h_\la&=&\frac{(n-\la+2)!}{(n - 2\la + 3)!}(\la-1)!(n-2\la+2)!\nonumber\\
	&=&\frac{(n-\la+2)!(\la-1)!(n-2\la+2)!}{(n - 2\la + 3)(n - 2\la + 2)!}\nonumber\\
	&=&\frac{(n-\la+2)!(\la-1)!}{(n - 2\la + 3)},\nonumber 
\end{eqnarray}
and Eq. (\ref{dimSn}) that contains only one factor gives 
\begin{equation}
	d_\lambda=\frac{(n-\la+1)-(\la-1)+2-1}{2-1}=n-2\la +3.\nonumber 
\end{equation}
Finally the dimensions of each $\Hcal_\la$ is given by 
\begin{eqnarray}
	\label{dim} D_\lambda&=&f_\la d_\la=(n - 2\la +3)\frac{n!(n - 2\la + 3)}{(n-\la+2)!(\la-1)!}\nonumber\\
	&=&\frac{(n - 2\la +3)^2 n!}{(n-\la+2)!(\la-1)!}\nonumber\\
	&=&\frac{(n - 2\la +3)^2 n!}{(n-\la+2)(n-(\la-1))!(\la-1)!}\nonumber\\
	&=&\frac{(n - 2\la +3)^2 n!}{(n-\la+2)(n-(\la-1))!(\la-1)!}\nonumber\\
	&=&\frac{(n - 2\la +3)^2}{(n-\la+2)}\binom{n}{\la-1}.\label{qubitdim} 
\end{eqnarray}

\subsection{Random qubits states and $\lambda$ symmetry}

Any quantum state can be seen as a linear combination of the states $\ket{b_i}$. In other words, any quantum state is a superposition of the different $\la$ symmetries. For instance, the separable state $\ket{01}$ is the superposition of the two Bell states $\ket{01}+\ket{01}/\sqrt{2}$ and $\ket{01}-\ket{01}/\sqrt{2}$. Therefore, this state contains the same amount of symmetry and antisymmetry. To quantify this, the following quantity is defined: 
\begin{equation}
	\label{w} w_\la(\psi)= ||P_\la \ket{\psi}||^2=\bra{\psi}P_\la P_\la\ket{\psi}=\bra{\psi}P_\la\ket{\psi}. 
\end{equation}
We will call this quantity the {\em weight} of $\la$ symmetry. It is just the square of the norm of the components in the subspaces $\Hcal_\la$. The more a state is $\la$ symmetric, the bigger the weight is. Note that by definition $\sum_{\la \dashv n}w_{\la}(\psi)=1$. Calculating this quantity for arbitrary states would not give more insight. That is why we will evaluate it for random states and consider related statistical quantities. The first of this quantity will be the mean values $\mu_1=\mean{w_\la}$, where $\mean{\cdots}$ stands for the average over the uniform measure $d\psi$ on the units sphere $\sum_{i} |\psi_i|^2=1$
in the whole Hilbert space. From Eq. (\ref{w}), we get 
\begin{equation}
	\mu_1=\mean{\sum_{ij}\psi_i^*(P_\la)_{ij}\psi_j}=\sum_{ij}\mean{\psi_i^* \psi_j} (P_\la)_{ij},\label{wsuite}\\ 
\end{equation}
where the indices $i$ and $j$ label computational basis states and $(P_\la)_{ij}$ are the matrix element of $P_\la$ of in that basis. Terms of the form $\mean{\psi_i^* \psi_j}$ can be calculated by different ways like the diagrammatic method described in \cite{Aubert03}. Product of component of quantum states averaged over $d\psi$
\begin{equation}\label{average}
	\mean{\psi_{i_1}^* \cdots \psi_{i_k}^*\psi_{j_1}\cdots \psi_{j_k}}=\int \psi_{i_1}^* \cdots \psi_{i_k}^*\psi_{j_1}\cdots \psi_{j_k}\,d\psi
\end{equation}
are nonzero if and only if each $i$ index has a corresponding $j$ index with the same value. In other words, when the string $\{i_1,i_2,\cdots,i_k\}$ is a permutation of the string $\{j_1,j_2,\cdots,j_k\}$, those strings are the binary form of the indices $i$ and $j$, respectively. It is important to notice that such permutation acts on the set of $j$ indices and have nothing to do with the previously considered permutations that act on the qubits. After the indices condition is fulfilled, the average value from Eq. (\ref{average}) takes the simplified form 
\begin{equation}
	\mean{|\psi_{i_1}|^{2}|\psi_{i_2}|^{2} \cdots |\psi_{i_k}|^{2}}.\label{average2} 
\end{equation}
Its value is now dependent on the possible degeneracy of the indices $i$. When those degeneracies are taken into account, Eq. (\ref{average2}) takes the form 
\begin{equation}
	\mean{|\psi_{i_1}|^{2l_1}|\psi_{i_2}|^{2l_2} \cdots |\psi_{i_m}|^{2l_m}}\,\text{ with } \sum_{j=1}^{m}l_j=k,\label{average3} 
\end{equation}
with a value that does not depend on the indices $i$, but only on the powers $l$. It will be written $F(l_1,l_2,\cdots,l_m)$. In \cite{Aubert03}, a similar expression is considered for unitary matrices drawn uniformly with respect to the Haar measure of the unitary group. However, a given column of such a matrix happens to be random states uniformly distributed over $d\psi$. By consequence, results that appears in \cite{Aubert03} for random unitary matrix can be directly used for random states giving
\begin{equation}
	F(l_1,l_2,\cdots,l_m)=\frac{(N-1)!\,l_1!\,l_2!\,\cdots l_m!}{(N+k-1)!}.\label{F} 
\end{equation}
For instance, $\mean{\psi_i^* \psi_j}$ will only be non-zero if $i=j$, which then leads to the term 
\begin{equation}
	\mean{|\psi_i|^2}=F(1)=\frac{(N-1)!}{N!}=\frac{1}{N}.\nonumber 
\end{equation}

It is then straightforward to continue the calculation of $\mu_1$. From Eq. (\ref{wsuite}) we get 
\begin{eqnarray}
	\mu_1&=&\sum_{i}\frac{1}{N}(P_\la)_{ii}=\frac{1}{N}\Tr(P_\la)\nonumber\\
	&=&\frac{D_\la}{N}. 
\end{eqnarray}
Therefore, for uniform states, the average weight of $\la$-symmetry is just the ratio between the dimensions of the subspace and the whole Hilbert space, as it could have been intuitively expected.

\begin{figure*}
	\begin{center}\includegraphics[width=6cm]{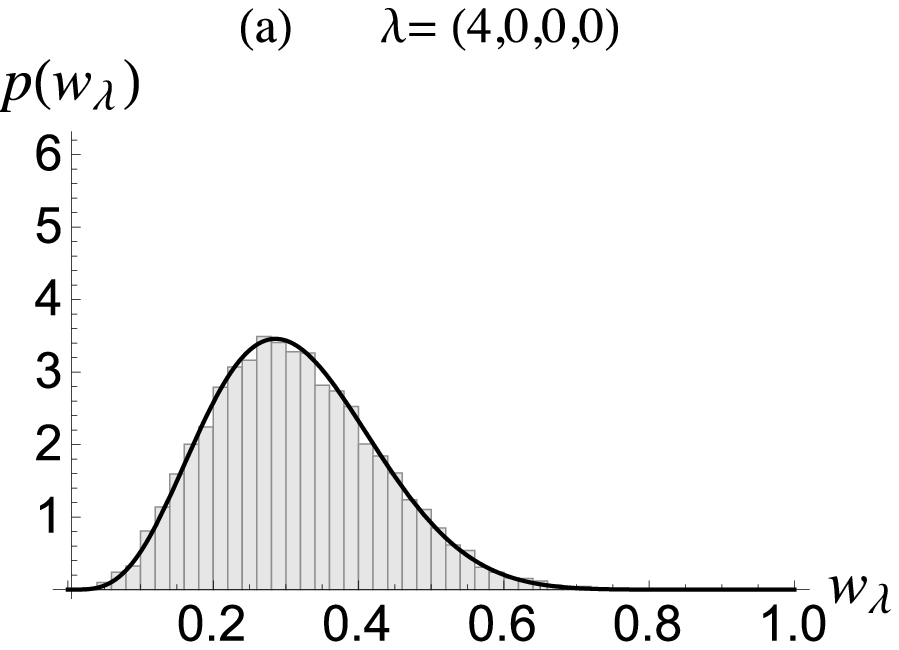}\includegraphics[width=6cm]{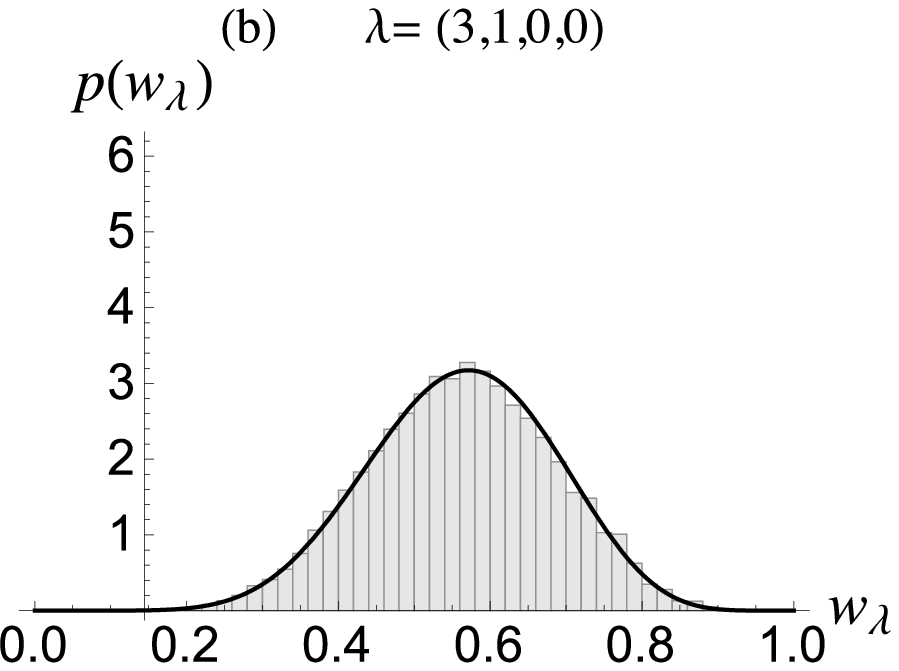}\includegraphics[width=6cm]{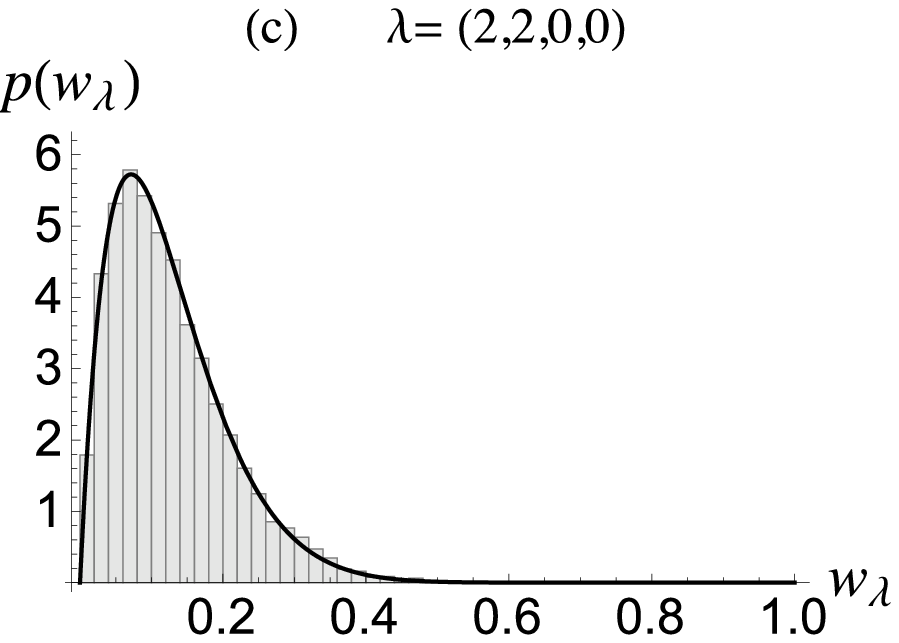}
	    \caption{Comparison of the analytical (curves) and numerical (histograms) distributions of $p(w_\la)$ for the three possible symmetries of four qubits. The numeric calculation used $10^4$ random states uniformly drawn with respect to the invariant measure of the unit sphere in the Hilbert space.}
	\label{fig1} 
	\end{center}
\end{figure*}

The second moment $\mu_2=\langle w_\la^2 \rangle$ can also be calculated analytically as 
\begin{eqnarray}
	\mu_2&=&\left\langle \sum_{i_1j_1}\psi_{i_1}^*(P_\la)_{i_1j_1}\psi_{j_1}\sum_{i_2j_2}\psi_{i_2}^*(P_\la)_{i_2j_2}\psi_{j_2} \right\rangle\nonumber\\
	&=&\sum_{i_1j_1i_2j_2}\langle\psi_{i_1}^*\psi_{i_2}^*\psi_{j_1}\psi_{j_2}\rangle (P_\la)_{i_1j_1}(P_\la)_{i_2j_2}.\nonumber 
\end{eqnarray}
The sum decomposes into three non-zero parts. First, the part where $i_1=j_1$ and $i_2=j_2$ but $i_1 \ne i_2$. Second, the part where $i_1=j_2$ and $i_2=j_1$ but $i_1 \ne i_2$. Finally, the part where $i_1=j_1=i_2=j_2$. It reads 
\begin{eqnarray}
	\mu_2&=&\sum_{i_1 \ne i_2}\mean{|\psi_{i_1}|^2|\psi_{i_2}|^2} (P_\la)_{i_1i_1}(P_\la)_{i_2j_2}\nonumber\\
	&+&\sum_{i_1 \ne i_2}\mean{|\psi_{i_1}|^2|\psi_{i_2}|^2} (P_\la)_{i_1i_2}(P_\la)_{i_2j_1}\nonumber\\
	&+&\sum_{i_1}\mean{|\psi_{i_1}|^4} (P_\la)_{i_1i_1}(P_\la)_{i_1i_1}. 
\end{eqnarray}
From Eq. (\ref{F}) we have 
\begin{eqnarray}
	\mean{|\psi_{i_1}|^2|\psi_{i_2}|^2} &=& F(1,1) = \frac{(N-1)!}{(N+1)!} = \frac{1}{N(N+1)}, \nonumber\\
	\mean{|\psi_{i_1}|^4} &=& F(2) = \frac{(N-1)!2!}{(N+1)!} = \frac{2}{N(N+1)}, \nonumber 
\end{eqnarray}
and then 
\begin{eqnarray}
	\mu_2&=&\frac{1}{N(N+1)}\Big(\sum_{i_1 \ne i_2}(P_\la)_{i_1i_1}(P_\la)_{i_2j_2}\Big.\nonumber\\
	&+&\Big.\sum_{i_1 \ne i_2}(P_\la)_{i_1i_2}(P_\la)_{i_2j_1}+2\sum_{i_1} (P_\la)_{i_1i_1}(P_\la)_{i_1j_1}\Big)\nonumber\\
	&=&\frac{1}{N(N+1)}\Big(\sum_{i_1 i_2}(P_\la)_{i_1i_1}(P_\la)_{i_2j_2}\Big.\nonumber\\
	&+&\Big.\sum_{i_1 i_2}(P_\la)_{i_1i_2}(P_\la)_{i_2j_1}\Big)\nonumber\\
	&=&\frac{1}{N(N+1)}\Big(\Tr(P_\la)Tr(P_\la)+Tr(P_\la^2)\Big)\nonumber\\
	&=&\frac{D_\la(D_\la+1)}{N(N+1)}.\label{w2} 
\end{eqnarray}
By induction, it is possible to show that the $k$th-moment $\mu_k=\mean{w_\la^k}$ writes 
\begin{equation}
	\mu_k=\frac{D_\la(D_\la+1)\cdots(D_\la+k-1)}{N(N+1)\cdots(N+k-1))}=\frac{(D_\la)_{(k)}}{(N)_{(k)}}\label{wk} 
\end{equation}
where $(x)_{(k)}=x(x+1)\cdots(x+k-1)$ are the Pochhammer symbols \cite{Erdelyi54,Abramowitz64}. The demonstration is detailed in the Appendix.

The fact that all the moments can be calculated is quite interesting. It gives a hope in the possibility to calculate the distribution of probability of the weights $p(w_\la)$ itself. In general, the collection of all the moments does not uniquely determine the distribution. However, if the moments satisfy some particular properties, such a uniqueness relationship exists. In \cite{Wu2002}, it is shown that if 
\begin{equation}
	\lim_{k \to \infty} \frac{1}{k}\left|\frac{\mu_{k+1}}{\mu_{k}}\right| 
\end{equation}
is finite, then the moments uniquely determine the probability distribution. From Eq. (\ref{wk}) we see that 
\begin{eqnarray}
	\lim_{k \to \infty} \frac{1}{k}\left|\frac{(D_{\la})_{(k+1)}(N)_{(k)}}{(N)_{(k+1)}(D_{\la})_{(k)}}\right|&=&\lim_{k \to \infty} \frac{1}{k}\left|\frac{D_\la+k}{N+k}\right|\nonumber\\
	&=&\lim_{k \to \infty}\frac{1}{k}\to0,\nonumber 
\end{eqnarray}
which is clearly finite. Therefore, the moment $\mu_k$ uniquely characterizes the probability distribution $p(w_\la)$ and it is worth trying to calculate it analytically from them. To do so, we need to introduce the characteristic function
\begin{eqnarray}
	\phi(t)=\mean{e^{itw}}&=& \sum_{k=0}^{\infty}\frac{(it)^k \mean{w^k}}{k!}=\sum_{k=0}^{\infty}\frac{(it)^k \mu_k}{k!}\nonumber\\
	&=&\int_{0}^{1} e^{itw}p(w)\,dw=\int_{-\infty}^{\infty} e^{itw}p(w)\,dw\nonumber\\
	&=&\TFm{p(w)}.\nonumber
\end{eqnarray}
\begin{figure*}
	\begin{center}
	\includegraphics[width=4.5cm]{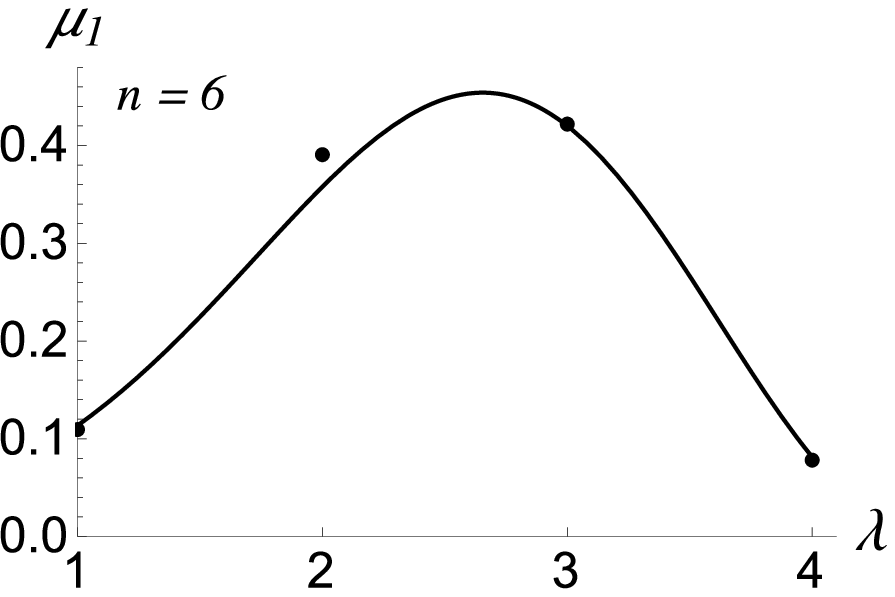}\includegraphics[width=4.5cm]{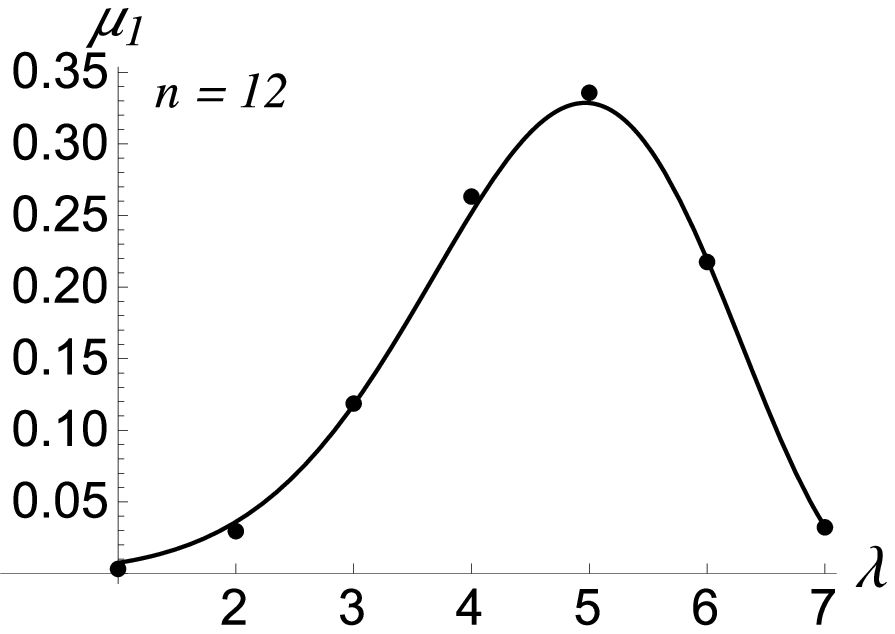}\includegraphics[width=4.5cm]{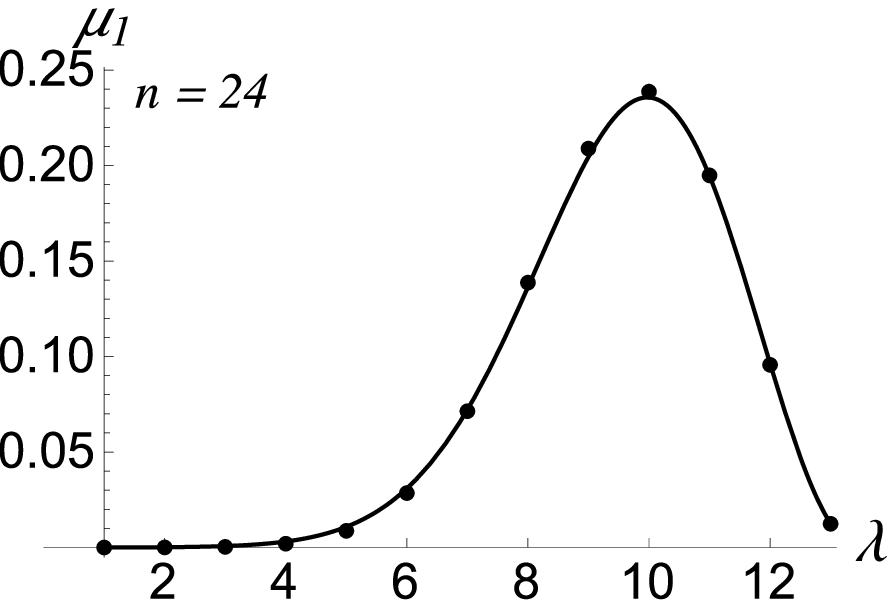}\includegraphics[width=4.5cm]{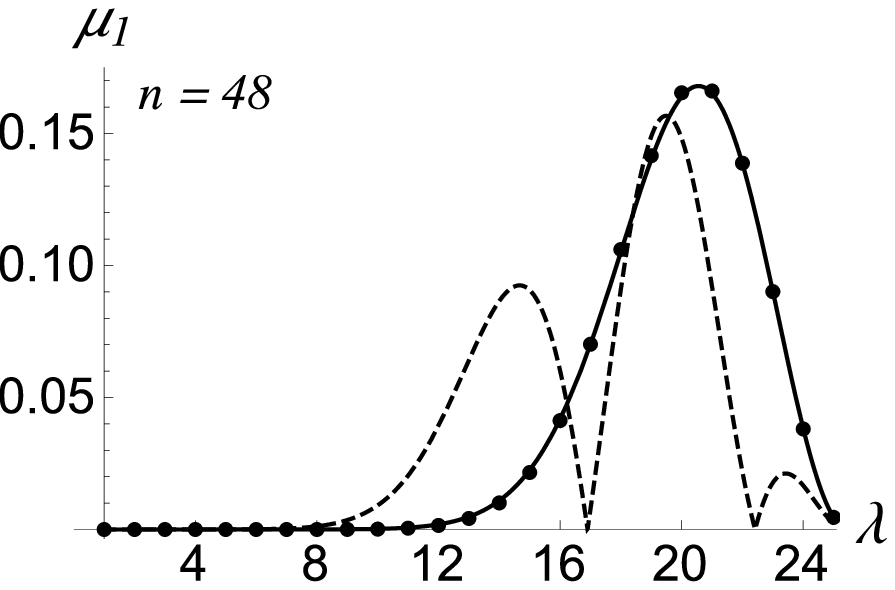}
	\caption{Comparison of the exact (dots) and approximate (solid curves) value of $\mu_1$ as a function of the representation label $\la$ for a different number of qubits. The dashed curve on the fourth figure represents 100 times the absolute value of the difference between the exact and the approximate value.}
\label{fig2} 
\end{center}
\end{figure*}
Note that, for simplicity, the index $\la$ is removed from the variable $w$. The characteristic function can be seen both as a series in the moments and as the inverse Fourier transform of the probability distribution. The strategy is the following: we first try to calculate the series, and then by taking its Fourier transform, we obtain the probability distribution. Knowing all the moments (\ref{wk}), the characteristic function written as a series 
\begin{equation}
	\phi(t)=\sum_{k=0}^{\infty}\frac{(it)^k}{k!}\frac{(D_\la)_{(k)}}{N_{(k)}}=_1^1F(D_\la,N,t)\label{serie} 
\end{equation}
happens to be the confluent hypergeometric function $\,_1^1F(a,b,it)$ with parameter $a=D_\la$ and $b=N$. The Fourier transform of such a function can be found in \cite{Erdelyi54} and finally we found the simple expression
\begin{eqnarray}
	p(w_\la)&=&\TF{\,_1^1F(D_\la,N,t)}\nonumber\\
	&=&\frac{w_\la^{D_\la}(1-w_\la)^{N-D_\la-1}}{B(D_\la,N-D_\la)}\,\forall w_\la \in [0,1], 
\end{eqnarray}
where the normalization factor $B(x,y)$ is the beta function defined for integer variables as $B(x,y)=\frac{(x-1)!(y-1)!}{(x+y-1)!}$. The calculated probability distribution is thus a {\em beta distribution}. Fig. \ref{fig1} compares this analytic distribution to the one calculated numerically by sampling $10^4$ random states. It is also interesting to consider such a distribution asymptotically. When the number of qubits goes to infinity, this distribution will be well approximated by a Dirac distribution centered in the mean value $\mu_1$. That is why $\mu_1$ is the moment that captures the most of the distribution and also deserves to be considered in the asymptotic limit. From (\ref{qubitdim}) and (\ref{w}), $\mu_1$ as a function of $n$ and $\la$ reads
\begin{equation}
	\mu_1=\frac{D_\la}{2^n}=\frac{(n - 2\la +3)^2}{2^n(n-\la+2)}\binom{n}{\la-1}.\nonumber\\
\end{equation}
In the asymptotic limit $n \to \infty$ the binomial coefficient can be approximated by a Gaussian function \cite{Abramowitz64} and then
\begin{eqnarray}\label{dimapprox}
\mu_1&\simeq&\frac{(n - 2\la +3)^2}{2^n(n-\la+2)}\frac{2^n}{\sqrt{\frac{1}{2}\pi n}}e^{-\frac{(\la-1-\frac{n}{2})^2}{\frac{n}{2}}}\nonumber\\
&\simeq&\frac{(n - 2\la +3)^2}{(n-\la+2)}\sqrt{\frac{2}{\pi n}}e^{-\frac{(\la-1-\frac{n}{2})^2}{\frac{n}{2}}}=\tilde{\mu_1}(\la)\label{mu1approx}
\end{eqnarray}
The quality of such an approximation can be observed in Fig. \ref{fig2} and the limit $n \to \infty$ is reached rapidly as $n$ increases. Especially on the figure corresponding to $n=48$, the dashed curve represents 100 times the absolute value of the difference between the exact and the approximate value. We see that even in the worst case, the error is less than $1\%$ of the exact value of $\mu_1$. We also notice that the profile of $\mu_1$ becomes more and more peaked around a given value of $\la$ that we can note $\la^*$. By using the approximation Eq. (\ref{mu1approx}), we can obtain the value of $\la^*$. The first derivative of $\tilde{\mu_1}(\la)$ with respect to $\la$ can be calculated and can be written in the form

\begin{equation}
	\frac{\partial \tilde{\mu_1}(\la)}{\partial \la}=A(\la)(a \la^3 + b\la^2 + c\la + d),
\end{equation}
where $A(\la)$ is a function that never vanishes on the range of $\la$ and
\begin{eqnarray}
	a&=&  -8 \nonumber\\
	b&=&  16 n + 36 \nonumber\\
	c&=&  -10 n^2 - 44 n - 52 \nonumber\\
	d&=&  2 n^3 + 11 n^2 + 27 n + 24.\nonumber
\end{eqnarray}
Solving such a cubic equation is straightforward on a computer even if the exact result is quite complicated to write as a function of $n$. When $n \to \infty$, the value of $\la^*$ behaves like
\begin{equation}
	\la^*(n)\simeq 0.49 n - 17.96,
\end{equation}
which means that it is really close to the biggest value of $\la$ as $n$ increases. 

All this statistical study shows is that in the asymptotic limit uniform random states are almost as antisymmetric as they could be, which is really close to the possible maximal value $\floor{\frac{n}{2}}+1$. This result is similar to those studying the amount of entanglement in random quantum states \cite{Cappellini2006,Dahlsten2014}.

\section{$\la$ symmetries and entanglement}
Recently, a lot of articles studied symmetric states, especially focusing on their multipartite entanglement and their non-locality \cite{Bastin2009,Martin2010,Aulbach2010,Markham2011,Baguette2014,Wang2012,Wang2013,Giraud2015,Arnaud2013}. In our context, those states are the ones contained in the subspace $\Hcal_{(n,0,\cdots,0)}$. They are really convenient to work with both analytically and numerically. This is mainly due to the fact that the dimension of $\Hcal_{(n,0,\cdots,0)}$ scales linearly with the number of qubits as it can be checked easily from Eq. (\ref{qubitdim}). For qubits, we found $D_{(n,0,\cdots,0)}=n+1$. Unfortunately, the power of such states is also their weakness. The complete permutational symmetry is quite constraining and a lot of known maximal entangled states do not satisfy it. Another inconvenience of this class of states is that they are locally equivalent to states that are not contained in the symmetric subspace. That is why considering the other subspace $\Hcal_\la$ makes sense to study.

\subsection{Extension of the Majorana representation}
When dealing with symmetric states, it is common to use the Majorana representation \cite{Majorana1932}. This representation has been used intensively in recent papers \cite{Bastin2009,Martin2010,Aulbach2010,Markham2011,Baguette2014,Wang2012,Wang2013,Giraud2015,Arnaud2013}. With our notations, this representation simply consists of the association of a product state $\ket{\Phi}=\ket{\phi_1}\cdots\ket{\phi_n}$ to a (in general entangled) state $P_{(n,0,\cdots,0)}\ket{\Phi}$. It is natural to extend the construction to the other $\la$-subspaces. Starting with a product state $\ket{\Phi}$, we construct all the states of the form $P_{\la}\ket{\Phi}$. Note that, in such constructions involving projections, normalization of the projected state is implicit. As we have already noticed, the product state $\ket{\Phi_2}=\ket{01}$ will lead to two Bell states when projected in the symmetric and anti-symmetric subspaces. For three qubits, the projection of the state $\ket{\Phi_3}=\ket{001}$ in the symmetric subspace gives the $\ket{W}$ state that maximizes the geometric entanglement \cite{Wei2003,Aulbach2010}.

The four qubits case is quite interesting. Let us consider the product state
\begin{equation}
	\ket{\Phi_4}=\ket{0}(\alpha\ket{0}+\beta\ket{1})(\alpha\ket{0}+\omega\beta\ket{1})(\alpha\ket{0}+\omega^2\beta\ket{1}),\nonumber
\end{equation}
with $\alpha=-1/3$, $\beta=2\sqrt{2}/3$, and $\omega=e^{2\pi i/3}$. It is straightforward to check that the projection of this state on $\Hcal_{(4,0,0,0)}$ and $\Hcal_{(2,2,0,0)}$ gives, respectively,
\begin{eqnarray}
P_{(4,0,0,0)}\ket{\Phi_4}&\to&\ket{T},\nonumber\\
P_{(2,2,0,0)}\ket{\Phi_4}&\to&\ket{HS}.\nonumber
\end{eqnarray}ß
Those two states are known to be the states that maximize different measures of entanglement in the whole Hilbert space \cite{Gour2010}. Note that, in \cite{Gour2010}, a state locally equivalent to $\ket{T}$ is actually considered. $\ket{T}$ is also known to maximize the geometric entanglement in the symmetric subspace \cite{Aulbach2010}. It is interesting to observe that there exists such a product state $\ket{\Phi_4}$ being the superposition of states with high entanglement content. For bigger sizes, it is not easy to pick relevant product states that lead to high entanglement projections. Investigating the extension of the Majorana representation would deserved an entire study. However, it is not the focus of this article. Instead, a numerical approach to observe what is happening for bigger size will be considered in the next part.

\subsection{Numerical approach}
In this part, we will again to calculate numerically relevant entanglement measures over random ensembles of quantum states. We have to keep in mind that such numerical calculations are limited because constructing projectors according to Eq. (\ref{proj}) involves a factorial number of operations. Such calculations can be performed on a laptop computer for a decent amount of times for up to eight qubits. As a measure of multipartite entanglement, we will use the generalization of the Meyer-Wallach measure \cite{Meyer2002} defined according to \cite{Scott2004} as
\begin{equation}
	Q_m=\frac{2^n}{2^n-1}\Big(1-\binom{n}{m}^{-1}\sum_{|A|=m} \Tr(\rho_A^2)\Big),\label{MW}
\end{equation}
where $\rho_A$ is the reduced density matrix of subsystem $A$ of size $m=1,2,\cdots,\floor{n/2}$ and where the sum is performed over all the subsystems of a given size $m$. For any $m$, this measure is zero for product states and maximal ($Q_m=1$) for perfect maximally multipartite entangled
states, even if reaching this maximum is known to be impossible for $n \ge 8$ \cite{Scott2004}.
\begin{table}
\begin{tabular}{| c | c | c | c | c |}
\hline
  n & Type of states & $\mean{\vec{Q}}$ & $\vec{Q}_{max}$ \\
\hline
\hline
 2  & Bell states & (1) &\\
    & Random states & (0.383) & (0.938)\\
    & $\ydiagram{2}$ & (0.505) & (1.00) \\
    & $\ydiagram{1,1}$ & (1.00) & (1.00)\\
\hline
 3  & $\ket{GHZ}$ & (1) & \\
    & $\ket{W}$ & (8/9) & \\
    & Random states & (0.640) & (0.897) \\
    & $\ydiagram{3}$ & (0.656) & (0.986) \\
    & $\ydiagram{2,1}$ & (0.657) & (0.815) \\
\hline
  4 & $\ket{HS}$ & (1.000,8/9) & \\
    & Random states & (0.811,0.692) & (0.926,0.807)\\
    & $\ydiagram{4}$ & (0.736,0.593) & (0.943,0.825) \\
    & $\ydiagram{3,1}$ & (0.852,0.705) & (0.938,0.832)  \\
    & $\ydiagram{2,2}$ & (1.000,0.726) & (1.000,0.884)\\
\hline
  5 & $\ket{M_5}$ & (1,1,1) & \\
    & Random states & (0.906,0.844) & (0.969,0.913) \\
    & $\ydiagram{5}$ & (0.791,0.659) & (0.957,0.853)\\
    & $\ydiagram{4,1}$ & (0.914,0.861) & (0.988,0.921)\\
    & $\ydiagram{3,2}$ & (0.870,0.788) & (0.943,0.858)\\
\hline
 6 & $\ket{M_6}$ & (1,1,1) & \\
   & Random states & (0.952,0.922,0.860) & (0.979,0.948,0.889)\\
   & $\ydiagram{6}$ & (0.817,0.691,0.625) & (0.964,0.848,0.776)\\
   & $\ydiagram{5,1}$ & (0.949,0.917,0.857) & (0.983,0.945,0.890)\\
   & $\ydiagram{4,2}$ & (0.947,0.896,0.832) & (0.980,0.935,0.868)\\
   & $\ydiagram{3,3}$ & (1.000,0.895,0.825)& (1.000,0.925,0.876) \\
\hline
\end{tabular}
\caption{Average and maximum value of $Q_m$, represented as the vectors $\mean{\vec{Q}}$ and $\vec{Q}_{max}$, are calculated using a sample of $10^3$ random states. ``Random states" indicates that the states are uniformly picked in the whole Hilbert space. Each Young diagram indicates that the states are uniformly picked in the subspace labeled by the corresponding diagram. Note that the left column is also filled with the analytical value for some states known as maximum of entanglement as $\ket{HS}$,  $\ket{M5}$ and  $\ket{M6}$ \cite{Laflamme1996,Higuchi2000,Brierley2007,Gour2010}.}
\end{table}
\begin{figure}
\begin{center}
\includegraphics[width=8.5cm]{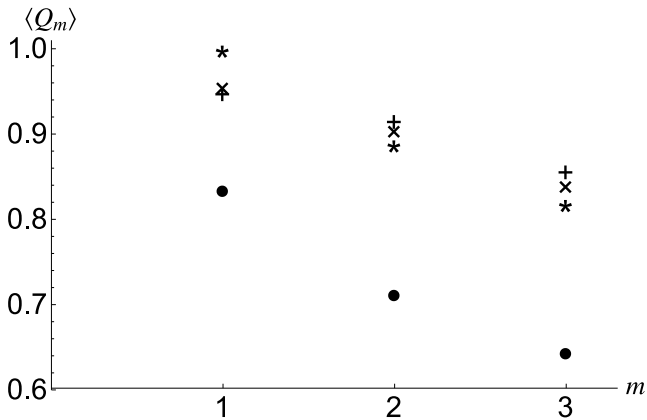}
\caption{For 6 qubits, comparison of the average values of $Q_m$ as a function of $m$ in the different $\la$-subspaces represented in the inverse lexicographical order by the markers $\bullet$, $+$, $\times$ and $*$, respectively.}
\label{fig3} 
\end{center}
\end{figure}
In Table I, the average and maximum values of $Q_m$ (represented as a vector with $\floor{n/2}$ components) are listed for quantum states of different sizes. For each size, we can compare those quantities for known maximally entangled states, random states picked uniformly in the whole Hilbert space (referred as ``random state"), and random states picked uniformly in the subspace $\Hcal_\la$ (referred by their corresponding Young diagrams). We notice that basically adding some antisymmetry to the states increase the multipartite entanglement on average and in maximum. The completely symmetric states appear to be the states with the lowest entanglement. We also notice that typical random states also contained a big amount of multipartite entanglement. It is due to the fact that on average random states are essentially antisymmetric (as $n$ goes to infinity) as observed in the previous part. Note that, in \cite{Scott2004}, an expression for the average of $Q_m$ for random states uniformly distributed over a given subspace is derived as a function of the projector in the subspace. In our context, such an average value for states picked in the subspaces $\Hcal_\la$ is given by
\begin{eqnarray}
	\mean{Q_m}_{\la}&=&\frac{2^m}{2^m-1}\Big[1-\mathcal{N}^{-1}\sum_{|A|=m} \Tr_A(\Tr_B(P_\la)^2)\nonumber\\
	&+&\left.\Tr_B(\Tr_A(P_\la)^2\right)\Big]\label{Qmean}
\end{eqnarray}
where $\mathcal{N}=D_\la(D_\la+1)\binom{n}{m}$. Such an expression can be evaluated on a computer the only limitation being the calculation of the projectors $P_\la$. For six qubits, the calculation of those average values can be calculated and are represented on the Fig. \ref{fig3}. As observed with the statistical calculation, the symmetric subspace contains {\em on average} less entanglement than the subspaces that are more anti-symmetric.

The more interesting fact concerns the states contained in the most antisymmetric subspace $\Hcal_{(n/2,n/2)}$ when $n$ is even. For those states up to the numerical precision, $Q_1$ seems to always reach its maximum. It means that those states are such that their 1-qubit reduced states are all maximally mixed. Note that states following such a property are called {\em 1-uniform} in \cite{Scott2004} or {\em 1-MM} in \cite{Arnaud2013}. It is an interesting fact because it is known that all entanglement monotones reach their maximum on the family of states that are 1-uniform \cite{Verstraete2003}. By definition, those states are also quantum error correcting code robust to any 1-qubit error. It is possible to create the projector $P_{(4,4)}$ in about 7 hours on a laptop computer and to check that all states in $\Hcal_{(4,4)}$ also satisfy $Q_1=1$ as observed in the most antisymmetric subspaces for $n=2,4$ and $6$. Therefore, it is worth trying to demonstrate analytically that this fact is general for any even $n$.

\subsection{Analytical approach}
In this part, we will be interested in the connection between multipartite entanglement and antisymmetry of a quantum state. It is illustrated by the following theorem:
\begin{mytheorem}
For an even number of qubits, all pure states contained in the most antisymmetric subspace are 1-uniform.
\end{mytheorem}

\begin{proof}
Let us first write the most antisymmetric partition $(n/2,n/2)$ as $\las$. By definition the basis states $\ket{b^{\las}_k}$ of $\Hcal_{\las}$ are linear combinations of computational states $\ket{i}$ with $i$ having a binary form that contains as much 0 than 1.
It is related to the fact that, in this situation, it always exists permutations that lead to the Young diagram filling $\dhalf$ that correspond to non zero-antisymmetrization. Those states can be written in a compact form as
\begin{equation}
\ket{b^{\las}_k}=\sum_{i_1 \cdots i_{n/2}}c_{\mathbf i}\ket{i_1 \cdots i_{n/2} \bar{i}_1 \cdots \bar{i}_{n/2}},\nonumber
\end{equation}
where the bar on the indices indicates binary complement, i.e., $\bar{0}=1$ and $\bar{1}=0$ and for some coefficients $c_{\mathbf i}=c_{i_1 \cdots i_{n/2} \bar{i}_1 \cdots \bar{i}_{n/2}}$. With the ``barred" notation the action of the Pauli matrix are given by
\begin{eqnarray}
	\sigma_x\ket{i}&=&\ket{\bar{i}},\nonumber\\
	\sigma_y\ket{i}&=&\mathrm{i}(-1)^{i}\ket{\bar{i}},\nonumber\\
	\sigma_z\ket{i}&=&(-1)^{i}\ket{\bar{i}},\nonumber\\
\end{eqnarray}
Let us now consider the ``rotated" states $\sigma_x^{(k)}\ket{b^{\las}_k}$, where $\sigma_x^{(l)}$ indicates that a Pauli-$x$ matrix is applied on the $l$th qubits.  Without lost of generality, by fixing $l=1$ the rotated states take the form
\begin{equation}
\sigma_x^{(1)}\ket{b^{\las}_k}=\sum_{i_1 \cdots i_{n/2}}c_{\mathbf i}\ket{\bar{i}_1 \cdots i_{n/2} \bar{i}_1 \cdots \bar{i}_{n/2}}.\nonumber
\end{equation}
It is a linear combination of computational states with a binary form that contains one extra 0 or 1 making the Young diagram filling unbalanced. In other words, $\sigma_x^{(k)}\ket{b^{\las}_k}$ is rotated in a subspace orthogonal to $\Hcal_{\las}$. A similar argument can be applied to a rotation performs by a $y$-Pauli matrix $\sigma_y^{(k)}$. The argument is similar for the rotation performed by the $z$-Pauli matrix $\sigma_y^{(k)}$ excepted that in this case the rotated states read
\begin{equation}
\sigma_z^{(1)}\ket{b^{\las}_k}=\sum_{i_1 \cdots i_{n/2}}(-1)^{i_1}c_{\mathbf i}\ket{i_1 \cdots i_{n/2} \bar{i}_1 \cdots \bar{i}_{n/2}},\nonumber
\end{equation}
which is also orthogonal to all the state $\ket{b^{\las}_k}$. As a consequence, all the scalar products $\bra{b^{\las}_k}\sigma_\alpha^{(l)}\ket{b^{\las}_k}$ are always zero ($\alpha=1,2$ or $3$) implying that all average values $\bra{\psi_{\las}}\sigma_{\alpha}\ket{\psi_{\las}}$ vanish for any state $\ket{\psi_{\las}} \in \Hcal_{\las}$. Finally, in \cite{Arnaud2013}, it is explained that an $n$-qubit state $\ket{\psi}$ is 1-uniform if and only if it satisfies
	\begin{equation}
		\bra{\psi}\sigma_{\alpha}^{(l)}\ket{\psi}=0.\label{1uniform}
	\end{equation}
for all $l\in[0,n]$ and $\alpha \in [1,3].$ 
\end{proof}
It is interesting to notice that this result looks like the opposite of theorem 1 in \cite{Arnaud2013} that says that {\em all pure states of qubits contained in the most symmetric subspace are as best 1-uniform}.

\section{Conclusion}
We investigated the $\la$-symmetries of a system of qubits in the context of quantum information. We proved that, on average, pure states of qubits picked at random with respect to the uniform measure on the unit sphere of the Hilbert space are almost as much antisymmetric as they are allowed to be. We then observed that multipartite entanglement, which is measured by the generalized Meyer-Wallach measure, tends to be larger in subspaces that are more antisymmetric than the complete symmetric one. Eventually, we proved that all states contained in the most antisymmetric subspace are relevant multipartite entangled states in the sense that their 1-qubit reduced states are all maximally mixed.

Following this present study, several research directions are possible. For instance, it would be interesting to find a more efficient way to construct projector Eq. (\ref{proj}) and perform numerical calculations for bigger $n$. Similarly, one could try to analytically calculate the average value of the measure $Q_m$ according to Eq. (\ref{Qmean}). One could also try to observe both analytically and numerically how the quantities $w_\la$ evolve when a quantum state is transformed by local unitaries.

\section*{Acknowledgments} I would like to thank Daniel Braun and Alicia Hartgrove for having read carefully the manuscript and for their useful comments. I would also like to thank Laure Benhamou and Ghyslain Protoy for their support.
\begin{widetext}
	\section*{Appendix} 
	\subsection*{Frobenius’ formula} The Frobenius’ formula \cite{Goodman2009} allows us to calculate $\chi_{\la}(\rho)$, the character in the representation $\lambda$ of each permutation that belongs to the conjugacy class $\rho$. To do so, the following quantities need to be introduced: 
\begin{itemize}
\item[-] The independent variables $x_1$, $\cdots$, $x_k$ where $k$ is at least as large as the last non-zero entry in the partition $\lambda$. 
\item[-] The power sums $P_j(x)=x_1^j+\cdots+x_k^j$ with $1\le j\le n$. 
\item[-] The discriminant $\Delta (x)=\prod_{i<j}(x_i-x_j)$. 
\item[-] The indices $l_i=\lambda_i+k-i$. 
\item[-] The polynomial $Q_{\la}^{\rho}(x)=\left[\Delta (x) \prod_{j=1}^{n}P_j(x)^{\rho^j}\right]$. 
\end{itemize}
Once the polynomial $Q_{\la}^{\rho}(x)$ is constructed, the Frobenius’ formula simply states that 
\begin{equation}
	\chi_{\la}(\rho)=\text{coefficient of } x_1^{l_1}\cdots x_k^{l_k} \text{ in } Q_{\la}^{\rho}(x). 
\end{equation}
Other formulas to calculate those characters exist, like the {\em Murnaghan-Nakayama rule} \cite{Loehr2010}.
	
\subsection*{$k$th moment of the distribution of $\mu_k$}
	
The $k$th moment $\mu_k$ of the distribution $p(w_\la)$ is given by 
\begin{equation}
	\mu_k=\sum_{\substack{i_1 \cdots i_k\\j_1 \cdots j_k}}\langle\psi_{i_1}^*\cdots\psi_{i_k}^*\psi_{j_1}\cdots\psi_{j_k}\rangle (P_\la)_{i_1j_1}\cdots(P_\la)_{i_kj_k}=\sum_{\pi \in S_k}\sum_{i_1 \cdots i_k}\langle|\psi_{i_1}|^2\cdots|\psi_{i_k}|^2\rangle (P_\la)_{i_1i_{\pi(1)}}\cdots(P_\la)_{i_ki_{\pi(k)}}.\nonumber\\
\end{equation}
Our goal is to express $\mu_k$ as a function of $\mu_{k-1}$. To do so, we will explicitly perform one of the sums, the one on the index $i_k$. Note that this choice is arbitrary. In the same vein, we will reduce the order of the mean value of the $\psi_i$. It is possible because those mean values do not depend on the indices $i$. From the expression in Eq. (\ref{F}), such a reduction gives 
\begin{equation}
	\langle|\psi_{i_1}|^2\cdots|\psi_{i_k}|^2\rangle=\frac{\langle|\psi_{i_1}|^2\cdots|\psi_{i_{k-1}}|^2\rangle}{N+k-1}.\nonumber 
\end{equation}
To perform the sum on $i_k$, it is judicious to split the symmetric group $S_k$ that acts on the indices $i$ in $k$ disjoined sets $S^{(j)}$ 
\begin{equation}
	S^{(j)}=\{\pi \in S_k \mid \pi(j)=k\}\, \forall j\in[1,k]. 
\end{equation}
By definition $\bigcup_{j}S^{(j)}=S_k$ and $\bigcap_{j}S^{(j)}=\emptyset$. Therefore, $\mu_k$ goes as 
\begin{eqnarray}
	\mu_k&=&\sum_{j=1}^k\sum_{\pi \in S^{(j)}}\sum_{i_1 \cdots i_k}\frac{\langle|\psi_{i_1}|^2\cdots|\psi_{i_k-1{}}|^2\rangle}{N+k-1} (P_\la)_{i_1i_{\pi(1)}}\cdots(P_\la)_{i_ki_{\pi(k)}}\nonumber\\
	&=&\frac{1}{N+k-1}\Big\{\sum_{j=1}^{k-1}\sum_{\pi \in S^{(j)}}\sum_{i_1 \cdots i_{k-1}}\langle|\psi_{i_1}|^2\cdots|\psi_{i_{k-1}}|^2\rangle (P_\la)_{i_1i_{\pi(1)}}\cdots[\sum_{i_k}(P_\la)_{i_ji_k}(P_\la)_{i_ki_{\pi(k)}}]\Big.\nonumber\\
	&+&\Big.\sum_{\pi \in S^{(k)}}\sum_{i_1 \cdots i_{k-1}}\langle|\psi_{i_1}|^2\cdots|\psi_{i_{k-1}}|^2\rangle (P_\la)_{i_1i_{\pi(1)}}\cdots[\sum_{i_k}(P_\la)_{i_ki_k}]\Big\}.\nonumber
\end{eqnarray}	
From the definition of the projectors $P_\la$, the sum over $i_k$ gives 
\begin{eqnarray}
	&\,&\sum_{i_k}(P_\la)_{i_ji_k}(P_\la)_{i_ki_{\pi(k)}} = (P_\la)_{i_ji_{\pi(k)}},\nonumber\\
	&\,&\sum_{i_k}(P_\la)_{i_ki_k} = \Tr(P_\la)=D_\la.\nonumber 
\end{eqnarray}
The reduction from the sum over $i_k$ also implies that each set $S^{(j)}$ becomes a group $S_{k-1}$ acting on the symbols $1$ to $k-1$. It can be seen by just considering the cycle structure of the elements contained in a set $S^{(j)}$. By definition, those elements are built with cycles of the form $(\cdots k j \cdots)$ multiplied by all the possible other cycles built with the remaining symbols. After removing the symbol $k$ in the cycle because of the summation, the previous set becomes simply the combination of all cycles built from $k-1$ symbols i.e. the group $S_{k-1}$. Then, realizing that the $(k-1)$ terms in the sum over $j$ are all equal, it follows
\begin{eqnarray}
	\mu_k&=&\frac{1}{N+k-1}\Big((k-1)\sum_{\pi \in S_{k-1}}\sum_{i_1 \cdots i_{k-1}}\langle|\psi_{i_1}|^2\cdots|\psi_{i_{k-1}}|^2\rangle (P_\la)_{i_1i_{\pi(1)}}\cdots(P_\la)_{i_{k-1}i_{\pi(k-1)}}\Big.\nonumber\\
	&+&\Big.D_\la\sum_{\pi \in S_{k-1}}\sum_{i_1 \cdots i_{k-1}}\langle|\psi_{i_1}|^2\cdots|\psi_{i_{k-1}}|^2\rangle (P_\la)_{i_1i_{\pi(1)}}\cdots(P_\la)_{i_{k-1}i_{\pi(k-1)}})\Big).\nonumber 
\end{eqnarray}
We can then identify the expression of $\mu_{k-1}$ 
\begin{equation}
	\mu_k=\frac{1}{N+k-1}\Big((k-1)\mu_{k-1}+D_\la\mu_{k-1}\Big)=\frac{(D_\la+k-1)}{(N+k-1)}\mu_{k-1}.\nonumber 
\end{equation}
By induction, using Eq. (\ref{w}) $\mu_1=D_\la / N$, it proves the result of Eq. (\ref{wk})
\begin{equation}
	\mu_k=\frac{(D_\la)_{(k)}}{(N)_{(k)}}.\nonumber
\end{equation}
\end{widetext}


\end{document}